\documentclass[aps,pra,reprint,superscriptaddress,amsmath,amssymb,floatfix]{revtex4-2}

\usepackage{graphicx}
\usepackage{xcolor}
\usepackage{amsthm}
\usepackage[expansion=false]{microtype}
\usepackage{hyperref}
\usepackage{orcidlink}

\hypersetup{
    colorlinks=true,
    linkcolor=blue,
    filecolor=magenta,
    urlcolor=blue,
    citecolor=blue,
    pdftitle={Computable Countermarkets and the Limits of Universal Trading},
    pdfauthor={Karl Svozil},
    pdfsubject={Computability, randomness, and limits of universal trading},
    pdfkeywords={trading, diagonalization, Halting Problem, rule inference, Martin-Lof randomness, algorithmic information theory, Busy Beaver, time reversal, market impact}
}

\theoremstyle{plain}
\newtheorem{theorem}{Theorem}

\newtheorem{corollary}[theorem]{Corollary}
\newtheorem{proposition}[theorem]{Proposition}

\theoremstyle{definition}
\newtheorem{definition}{Definition}
\newtheorem{criterion}{Criterion}
\theoremstyle{remark}
\newtheorem{remark}{Remark}

\begin{document}

\title{Computable Countermarkets and the Limits of Universal Trading}

\author{Karl Svozil\,\orcidlink{0000-0001-6554-2802}}
\affiliation{Institute for Theoretical Physics, TU Wien, Wiedner Hauptstra{\ss}e 8-10/136, A-1040 Vienna, Austria}
\email{karl.svozil@tuwien.ac.at}
\homepage{http://tph.tuwien.ac.at/~svozil}

\date{October 9, 2026}

\begin{abstract}
We explain why no trading algorithm can guarantee profit in every market.
For each deterministic program that always returns an exact rational position,
we construct a fixed, algorithmically generated price path on which every active
position loses and inactivity earns nothing. This holds with positive,
continually changing prices, costless trading, and unlimited computation time.
Conversely, a supplied total computable price program with nonzero moves
admits a computable trader that profits at every step. Mutual countering
depends on information, timing, and termination.
Approximation-based real outputs preserve the cumulative obstruction but
not the loss on every active position. Separate arguments limit learning
market rules, certifying future events, and establishing randomness from
finite data. Betting-martingale characterizations connect trading with
computable and Martin-L\"of randomness.
Useful strategies may exploit market structure, information, or compensation
for risk, while benchmark performance need not imply profit. Reversing and
rearranging price histories within the assumed market class provide practical
stress tests, distinguishing conditional success from universal guarantees.
\end{abstract}

\maketitle

\section{Introduction}
\label{sec:intro}

Throughout this paper the elementary configuration has two sides: a
\emph{trader} and a \emph{market}. At each date the trader converts observed
history into a position, meaning its chosen asset holdings; the market supplies
the next price. Depending on the question, the market side is a realized path
(a sequence of prices), a path-generating program, a probability law assigning
probabilities to possible outcomes, or a price process constrained by
no-arbitrage. A price process models prices as random variables over time;
no-arbitrage excludes a risk-free opportunity to gain from zero initial
capital~\cite{Harrison1981,Delbaen1994}.
Holding the trader fixed while changing the admitted market side reveals which
universal claim is at issue. A pathwise claim concerns individual price paths,
rather than averages over a probability law.

A strategy is \emph{universal} here only if, after discounting by a specified
numeraire, it produces strictly positive terminal gain from zero initial
capital on every admitted trajectory at a specified finite horizon. The
numeraire is a positive-valued reference asset, here a bank account;
discounting expresses prices and gains in units of that asset. The horizon is
the final evaluation date, and a trajectory is a price path. This claim differs
from nonnegative gain everywhere with positive gain somewhere, eventual profit
at a path-dependent date, expected profit (a probability-weighted average), and
small regret to a benchmark. A benchmark is a specified reference investment
or strategy; regret measures a strategy's performance shortfall relative to it.
Zero cost means zero initial investment, self-financing means no later external
cash injections, and admissibility imposes restrictions on allowable trading,
such as a lower bound on wealth. These conditions distinguish trading gain
from positive wealth obtained by investing positive capital. The order of the
quantifiers, which specify ``there exists'' and ``for every,'' is
\begin{equation}
  \exists\,\sigma\ \forall\,\mathcal E:
  \qquad \sigma\ \hbox{wins against }\mathcal E.
  \label{eq:universalquantifiers}
\end{equation}
Here \(\sigma\) denotes a trading strategy and \(\mathcal E\) an admitted
market; \(\exists\) means ``there exists'' and \(\forall\) means ``for every.''
``Wins'' means earning strictly positive discounted gain from zero initial
capital at the specified finite horizon. The order requires one fixed strategy
to succeed against every admitted market.
By contrast, diagonalization---constructing a counterexample tailored to the
proposed strategy---establishes
\begin{equation}
  \forall\,\sigma\ \exists\,\mathcal E_\sigma:
  \qquad \sigma\ \hbox{does not win against }\mathcal E_\sigma.
  \label{eq:counterquantifiers}
\end{equation}
The subscript in \(\mathcal E_\sigma\) indicates that the countermarket may
depend on the chosen strategy; there need not be one market defeating all
strategies.

A computable rule is one executable by an algorithm; ``effective'' and
``recursive'' have this same computational meaning here. A total rule returns
an answer on every allowed input, and a deterministic rule makes no random
choices. For a fixed total deterministic computable trader, its code can be
embedded in a total program that recursively chooses the next bounded proportional price
move opposite to each nonzero position and a nonzero move during inactivity.
The resulting countermarket is one
fixed computable path; any finite prefix can be generated offline. A trader that
never takes a position earns zero, which already refutes a guarantee of strict
profit. This is the trading counterpart of the next-bit construction
\(x_0=1-F(\varnothing)\) and
\(x_t=1-F(x_0,\ldots,x_{t-1})\) for \(t\geq1\): it defeats the fixed proposed
rule, not every rule on one exceptional sequence. Here a bit is either zero
or one, \(F\) predicts the next bit from the preceding bits,
\(\varnothing\) denotes the empty history, and \(x_t\) is the bit chosen
opposite to that prediction at date \(t\). A prefix is an initial finite
segment of a sequence; generating it offline means computing it without
interacting with a live market.

The opponent-specific character of the construction also permits a converse.
Given a total computable price program with nonzero changes, a trader can
compute the next change and choose its profitable sign.
Section~\ref{sec:countertrader} establishes this countertrader result and
distinguishes separate counterconstructions from mutually consistent
responses in a single interaction. Code access does not give both sides
the ability to respond last. Section~\ref{sec:martingales} connects
bounded-credit trading to betting martingales and distinguishes
computable from lower-semicomputable capital processes.

Three complementary limits use different market information. Gold's learner
sees an accumulating history and must eventually stabilize on a correct
program; no effective learner identifies every binary total recursive time
function (a program assigning a bit to every date), although restricted
classes such as the primitive-recursive time functions are identifiable in
the limit~\cite{go-65,go-67}. Primitive-recursive functions are built from
elementary functions using composition and recursion whose number of stages is set by the input. Identification in the limit means eventually keeping
one correct rule, without necessarily knowing when that point has been reached.
A source-code procedure that settled every future event of Turing-universal
generators, which can simulate any program, would instead decide the Halting
Problem: whether an arbitrary program eventually stops. Busy-Beaver growth,
the growth of the longest halting runtime among programs up to a given length,
excludes a computable worst-case waiting horizon. Finally, relative to a
specified computable probability law, a Martin-L\"of-random record escapes
every effective null test: an algorithmically specified test whose full
failure set has probability zero. Such records have measure one, meaning
probability one, also called ``almost surely.'' Finite samples cannot certify
this infinite-sequence property.

Repeatable practical success must therefore identify structure absent from the
unrestricted market class: persistence (continuation of past behavior), mean
reversion (movement back toward a reference level), a risk premium (expected
compensation for bearing risk), a benchmark, lawful information, or another
distributional asymmetry (an imbalance in possible outcomes or their
probabilities). An edge is an expected performance advantage under specified
market conditions. Large, crowded, or public strategies may also change the
distribution through price impact, meaning price changes caused by their own
orders, and adaptation by other participants~\cite{Lo2004}. This falsifiable restriction,
one that evidence could contradict, is the operational meaning of ``insight''
used below.

\begin{table*}[t]
\caption{\label{tab:claimmap}Distinct claims organized by the trader--market
configuration. They are complementary, not interchangeable.
Here \(P_\sigma\) is the price path constructed for strategy \(\sigma\);
a black box reveals outputs but not its internal program, and a promise
restricts the inputs on which an algorithm must work.}
\centering
\begin{ruledtabular}
\begin{tabular}{p{0.18\textwidth}p{0.27\textwidth}p{0.24\textwidth}p{0.23\textwidth}}
\textbf{Mechanism} & \textbf{Market side} & \textbf{Formal scope}
& \textbf{Conclusion} \\
\colrule
Diagonal\newline\mbox{countermarket}
& Fixed total program constructed from the fixed trader's code
& Pathwise worst case:
\(\forall\sigma\,\exists P_\sigma\)
& No total deterministic computable trader with exact rational outputs wins
on every positive rational-valued computable price path \\
Known-program\newline countertrader
& Supplied total computable price program with nonzero changes
& Opponent-specific construction; every positive finite horizon
& Each such path admits a profitable computable trader; the two perfect
current responses cannot hold together \\
Gold identification\newline in the limit
& Accumulating computable history, or interaction with a black box
& Success is eventual stabilization on a correct program, without a known
stopping date
& No effective learner identifies every computable binary time function;
restricted classes can be identified \\
Halting reduction and\newline Busy Beaver
& Total program or generator, under an explicit promise
& Individual encoded computation over an unbounded horizon
& No general computable profit or market-event decider, or computable
description-size waiting horizon \\
Algorithmic\newline randomness
& Passive binary record under the computable fair-coin reference law
& Each Martin-L\"of-random path; such paths have measure one
& Independent fair trials yield Martin-L\"of-random records almost surely;
finite tests cannot certify randomness \\
Betting martingales
& Binary fixed-magnitude paths; bounded-credit trading
& Unbounded capital along an infinite record
& Computable and lower-semicomputable capital characterize different
randomness notions; finite-horizon success has a probability bound \\
\end{tabular}
\end{ruledtabular}
\end{table*}

The explicit diagonal countermarket is the constructive core of the paper.
Table~\ref{tab:claimmap} organizes the complementary limits into a taxonomy that
keeps distinct several statements commonly summarized as ``no universal
strategy.'' The construction and its proofs are expository reformulations
rather than a claim of a newly discovered impossibility theorem.
The analysis connects this quantifier map to the literature on unpredictable
prices~\cite{Bachelier1900,Samuelson1965,Fama1970}, Gold learning, computability
and randomness, time-reversal and permutation stress
testing (reversing or otherwise reordering a record), and the conditions under
which strategies may work ``for all practical purposes'' (FAPP, following
Bell~\cite{bell-a}).

\section{A Computable Countermarket for Every Trader}
\label{sec:countermarket}

The trader is a program \(M_\sigma\) mapping each observed finite price history
to a position; the market side supplies the price history.

Throughout the discrete constructions, dates \(t=0,1,2,\ldots\) count trading
periods, and \(T\) denotes a finite final date. The quantity \(P_t=S_t/B_t\) is the discounted
price of one asset, where \(S_t\) is its price in a fixed currency and
\(B_t>0\) is the value of one unit of the chosen risk-free bank account.
For example, if both initially equal \(100\), and later \(S_1=121\) and
\(B_1=110\), the discounted price rises from \(1\) to \(1.1\).
This measures performance relative to the bank account, not inflation.
The position \(w_t\) is a number of shares chosen before \(P_{t+1}\) is
observed. A positive position holds shares, a negative position sells borrowed
shares short, and zero means no asset exposure. A short position profits from
a price fall and loses from a rise. Trading is self-financing: changes in holdings are funded through
the existing share and bank balances, with no external cash injections.
Let \(V_t\) be total portfolio wealth at date \(t\), measured in bank-account
units. More explicitly, holding \(w_t\) shares and the residual
\(V_t-w_tP_t\) in the discounted bank account gives
\[
  V_{t+1}=V_t+w_t(P_{t+1}-P_t).
\]
We omit transaction costs, borrowing limits, leverage caps (limits on exposure
relative to capital), and solvency constraints (requirements to remain able to
meet obligations). The counterexample therefore also applies to any subclass of
these strategies satisfying additional restrictions.

Let
\(M_\sigma\colon\mathbb Q_{>0}^{<\omega}\to\mathbb Q\)
be a total computable function on finite rational histories, encoded in the
usual effective way, with \(w_t=M_\sigma(P_0,\ldots,P_t)\).
Here \(\mathbb Q\) is the set of rational numbers, \(\mathbb Q_{>0}\) its
positive members, and the superscript \(<\omega\) denotes all finite
sequences; the arrow specifies the function's input and output sets.
Totality requires an
action in finite computation time on every input; it supplies no prescribed
running-time bound or guarantee of meeting a trading deadline. Here the program outputs an exact rational representation, not merely a
procedure approximating a real number. Such outputs model finite-precision
orders and make comparison with zero decidable. Let \(\mathcal P_{\mathrm{comp}}^{+}\) denote the positive
rational-valued computable price paths.

\begin{samepage}
\begin{theorem}[Computable Diagonal Countermarket]
  \label{thm:counterpath}
  Fix a rational initial price \(P_0>0\) and a rational proportional price-change
  magnitude \(0<\epsilon<1\). For every total deterministic computable strategy \(M_\sigma\)
that returns exact rational positions, there exists a fixed path
  \(P^{M_\sigma}\in\mathcal P_{\mathrm{comp}}^{+}\), starting at \(P_0\), with
  \(P_{t+1}/P_t\in\{1-\epsilon,1+\epsilon\}\) at every date, on which its
  self-financing gain
  \[
    G_T=\sum_{t=0}^{T-1}w_t(P_{t+1}-P_t)
  \]
  is non-positive at every finite horizon \(T\). It is strictly negative after
  the strategy has taken a nonzero position at least once.
\end{theorem}
\end{samepage}
The superscript in \(P^{M_\sigma}\) labels dependence on the trader's program;
it is not a power. The gain \(G_T\) sums the position times the price change
over the first \(T\) periods and equals final discounted wealth when
\(V_0=0\). The proportional return is the price change divided by the
preceding price; its magnitude here is always \(\epsilon\).

\begin{proof}
Choose rational \(P_0>0\) and rational \(0<\epsilon<1\). Recursively compute
\(w_t=M_\sigma(P_0,\ldots,P_t)\) and set
\begin{equation}
P_{t+1}=
\begin{cases}
P_t(1-\epsilon),&w_t>0,\\
P_t(1+\epsilon),&w_t<0,\\
P_t(1+\epsilon(-1)^t),&w_t=0.
\end{cases}
\label{eq:adversary}
\end{equation}
Because \(M_\sigma\) is total and rational-valued, this recursion is a total
program for one fixed infinite price path. Its prices remain positive and
computable, and every return has magnitude \(\epsilon\). In each period,
\(w_t(P_{t+1}-P_t)\leq0\), with strict inequality when \(w_t\neq0\).
Summing proves the claim at every finite horizon.
\end{proof}

In the last case of Eq.~\eqref{eq:adversary}, the alternating market move is a
convenient but arbitrary choice. Since \(w_t=0\), any market move yields zero
gain in that period. Within the theorem's fixed-magnitude class, any computable
choice between returns \(+\epsilon\) and \(-\epsilon\) is admissible, allowing
a wide variety of market-move sequences in place of the alternating convention.

For a concrete example, fix the rule that holds one share for the first
period, shorts one share for the second, and then holds no shares. With
\(P_0=100\) and \(\epsilon=1/10\), Eq.~\eqref{eq:adversary} produces
\begin{align*}
  (w_0,w_1,w_2)&=(1,-1,0),\\
  (P_0,P_1,P_2,P_3)&=(100,90,99,108.9).
\end{align*}
The successive discounted gains are \(-10,-9,0\), so \(G_3=-19\), all in
bank-account units. The prices are constructed after fixing the rule; they
are illustrative, not market data.

\begin{remark}[Logical and Economic Scope of the Countermarket]
\label{rem:counterpath-scope}
Theorem~\ref{thm:counterpath} uses the trader's \emph{code}, not a physical
market's observation of a deployed order. Once \(M_\sigma\) is fixed, the
countermarket program is fixed and can generate any desired finite prefix
offline. It is therefore an admissible deterministic computable counterexample
within the stated model, although not necessarily an efficient one.
Computability and economic rationality are separate assumptions. If an
artificial opponent is assigned the objective of preventing this trader's
strict gain, the rule serves that objective; it is not thereby an economic
equilibrium or an arbitrage-free price process. Another trader may profit
on the resulting path and has, in turn, its own computable countermarket.
Section~\ref{sec:countertrader} makes this converse and its limits explicit.

Code access is an information assumption of the construction. Without it, a
market required to forecast the next buy-or-sell action from past actions
faces the direct next-bit obstruction of Section~\ref{sec:learning}: every
total computable predictor has a computable time-dependent trading rule
against which each forecast is wrong. Identifying the trader's generating
rule is a different task. Since time-dependent buy-or-sell rules can generate
every computable binary sequence, Gold's identification limit also applies
to that task~\cite{go-65,go-67}. Neither statement establishes the failure of
every antagonistic action. Unbounded questions about a trader, such as
whether it will ever take a nonzero position, can also encode the Halting
Problem and remain undecidable even when its code is supplied
(Section~\ref{sec:halting}). If the market instead observes the current
position before setting the next price, it can oppose that position directly
without knowing or inferring the strategy.
\end{remark}

This is a direct strategy-relative diagonal construction. The proposed
universal trader is used to compute an object outside its success set: long is
paired with down, short with up, and neutrality with an alternating nonzero
return. Just as
the complement sequence defeats a purported predictor of every computable
sequence, \(P^{M_\sigma}\) defeats the purported trader on every horizon.
Enumeration of all programs is unnecessary because the universal claim itself
supplies the program to be diagonalized against. The quantifiers are
\(\forall M_\sigma\,\exists P^{M_\sigma}\). A constant discounted price path
already gives every self-financing rule zero trading gain and refutes strict
profit on a class containing that path. The stronger content of
Theorem~\ref{thm:counterpath} is that every nonzero position of the selected
trader incurs a strict loss even on the class with fixed return magnitude
\(\epsilon\), which contains no flat path. It does not supply a path on which every strategy
strictly loses: permanently long and permanently short positions cannot both
lose on the same one-step price move.

\begin{remark}[Exact Outputs and Computable-Real Positions]
\label{rem:rational}
Theorem~\ref{thm:counterpath} assumes that each position is returned as an
exact rational number, whose sign is decidable. An alternative representation
provides, on input a rational price history and precision \(k\), a rational
approximation to the position with error at most \(2^{-k}\). The approximation
algorithm must be total and uniform in both inputs. Strict positivity and
strict negativity are then semidecidable, but deciding between them and zero
need not be possible. The distinction concerns the effective representation:
even positions whose values are all dyadic rationals can exhibit this problem.
A generic sign search may wait forever at an exactly zero position, so
Eq.~\eqref{eq:adversary} need not define a total program in this representation.
Finite-precision executed orders fit the exact-output model; the
approximation-based model is a mathematical idealization.
Empty price histories may be assigned position zero throughout.
\end{remark}

\begin{proposition}[Failure of Per-Period Opposition]
\label{prop:real-positions}
Fix \(P_0\) and \(\epsilon\) as in Theorem~\ref{thm:counterpath}, retaining
the same price-path class but allowing approximation-based outputs instead
of exact rational outputs.
There is a trader with uniformly computable-real positions, depending only
on the date, such that every computable path in the fixed-magnitude class
gives it a strictly positive gain in at least one period.
\end{proposition}

\begin{proof}
Enumerate the partial computable functions
\(\varphi_e\colon\mathbb N\rightharpoonup\{-1,+1\}\).
If \(\varphi_e(e)\) halts after \(n\geq1\) steps with output \(b\), set
\(w_e=b2^{-n}\); otherwise set \(w_e=0\).
Simulating for \(k\) steps and returning the exact value if it has appeared,
and zero otherwise, approximates \(w_e\) within \(2^{-k}\), uniformly in
\(e\). Use the history length minus one to obtain the date.
For any computable path in the stated class,
\(d_t=\operatorname{sgn}(P_{t+1}-P_t)\) is a total computable sign sequence.
Choose an index \(e\) with \(\varphi_e=d\). Then \(w_e\neq0\) has sign
\(d_e\), so \(w_e(P_{e+1}-P_e)>0\).
\end{proof}

The two results concern different representations of traders on the same
class of paths. The positions just constructed are dyadic rationals, but no
uniform algorithm returns their exact rational values or decides their signs
including zero. Code access therefore does not make
Eq.~\eqref{eq:adversary} executable for this trader. The proof covers every
computable path, including any produced from the trader's code; it makes no
independence assumption about the path. A positive gain in one period also
need not make cumulative gain positive.

\begin{proposition}[Cumulative Nonpositivity Survives]
\label{prop:real-cumulative}
For every trader whose positions on rational histories are uniformly
computable reals, there exists a computable path in the same fixed-magnitude
class with \(G_T\leq0\) at every finite horizon.
\end{proposition}

\begin{proof}
Consider the computable all-up path \(P_t=P_0(1+\epsilon)^t\).
If every position along it is zero, it is already the required path.
Otherwise let \(s\) be its first nonzero-position date. Retain its history
through \(P_s\), then choose the return opposite to that position.
This gives \(G_T=0\) for \(T\leq s\) and \(G_{s+1}<0\).
The finite prefix and this direction can be hard-coded.

Thereafter, given \(G_t<0\), the two possible next cumulative gains are
\[
  A_+=G_t+\epsilon P_tw_t,\qquad A_-=G_t-\epsilon P_tw_t.
\]
Both are computable reals and their average is \(G_t<0\); hence at least
one is strictly negative. Run their strict-negativity tests in parallel and
choose the direction whose test first certifies negativity. Each step
terminates and preserves \(G_{t+1}<0\), generating a computable infinite
price path. The initial case distinction and finite-prefix selection establish
existence, without providing a uniform procedure from trader code to
countermarket code.
\end{proof}

Thus approximated positions invalidate the guarantee that \emph{every}
active position loses, but not the existence of a computable countermarket
preventing positive cumulative gain at every horizon. Given the trader's
approximation program, a Halting-Problem oracle can decide each sign,
including zero, and recover the per-period construction relative to
\(\emptyset'\).

\section{Learning a Market Rule from History: Gold's Limit}
\label{sec:learning}

A trader normally receives neither the market generator's source code nor a
proof of its law. The trader sees an expanding record
\(x_0,x_1,\ldots\) and must infer from each finite prefix how it will continue.
Gold's function-identification model formalizes exactly this information flow.
At time \(n\), an effective learner \(L\), itself an algorithm, maps
\((x_0,\ldots,x_n)\) to the index, or numbered program description, of a
Turing machine, an abstract model of computation. It
\emph{identifies \(x\) in the limit} if, from some finite date onward, every
guess is the same index and that machine computes the whole sequence. The
learner need not know when its last change of mind has occurred
\cite{go-65,go-67}.

Let \(\mathrm{REC}_2\) denote the total computable binary sequences. Gold's
negative theorem for binary-valued time functions has the trader--market
quantifier order, with \(L\) ranging over effective learners,
\begin{equation}
  \begin{aligned}
  &\forall L\ \exists x\in\mathrm{REC}_2: \\
  &\quad L(x_0,\ldots,x_n)\text{ does not stabilize}\\
  &\quad \text{on a program for }x .
  \end{aligned}
  \label{eq:goldquantifiers}
\end{equation}
The adverse history depends on the learner. This is closely analogous to
\(\forall M_\sigma\,\exists P^{M_\sigma}\), although Gold's failure criterion
is failure to converge to a fixed correct program name, not financial loss.

The corresponding claim about literal next-move forecasting is sharper and
has the short direct diagonal proof already used above. If a total computable
predictor \(F\) outputs one bit from every finite binary history, the recursive
definition
\[
  x_0=1-F(\varnothing),
  \qquad
  x_t=1-F(x_0,\ldots,x_{t-1})\quad(t\geq1)
\]
produces a computable market record on which \(F\) is wrong at every date.
Thus the proposed statement---every total next-move predictor has a computable
countermarket---is correct, but it is this direct prediction diagonalization,
not the literal statement of Gold's identification theorem. Nonconvergence of
program names alone does not count the learner's next-step errors.

The opposite-bit construction also appears explicitly in Legg's analysis of
universal prediction~\cite{Legg2006}. The qualifier \emph{computable} is
essential. Solomonoff induction combines program-based predictive models,
giving shorter descriptions greater weight. It defines an uncomputable
universal predictive ideal; for computable probabilistic environments its conditional predictions
converge, with quantitative error bounds developed by Hutter
\cite{solomonoff,solomonoff-II,solomonoff1,Hutter2001}. It therefore evades,
rather than contradicts, the diagonal result: it is not a total executable
trading rule, and predictive convergence is not Gold's requirement of
stabilizing on one exact generating program. For probability-valued forecasts,
V'yugin gives a further, distinct impossibility: a randomized construction can
make any partial computable forecaster fail calibration with probability
arbitrarily close to one~\cite{Vyugin2006}. A partial algorithm may fail to
return an answer on some inputs; calibration requires forecast probabilities
to agree with observed frequencies on the corresponding forecast groups.
Calibration, next-bit error, and identification in the limit are three
different criteria.

Gold's results also show why the restriction on the market side matters. Every
uniformly effective indexed family of total computable functions---one
algorithm computes any member from its index and input---is identifiable in the limit: a learner selects the least index consistent with
the observed prefix, and eventually eliminates every smaller incorrect index.
In particular, primitive-recursive time functions form such a family. In Gold's interactive
black-box model, where the eventual hypothesis must reproduce future
input--output behavior, finite-state boxes, with only finitely many internal
states, are weakly identifiable in the limit, whereas the class of
primitive-recursive black boxes is not
\cite{go-65}. These are different object classes, not contradictory claims.
Gold's 1967 language model further makes explicit that learnability depends on
the admitted class, the presentation of the evidence, and the allowed names
\cite{go-67}. For trading, a forecast can therefore be justified only relative
to a specified restriction on possible markets.

\section{Universal Computation and the Halting Barrier}
\label{sec:halting}

Gold's learner works from observations and may revise its hypothesis
indefinitely. The next results ask a separate, stronger operational question:
can one terminating program inspect arbitrary encoded traders or market
generators and correctly decide their unbounded future behavior? The preceding
diagonal construction already disproves a trader that wins on every computable
path and does not invoke the Halting Problem. The source-code
event-certification question does. A problem is undecidable if no algorithm
can always terminate with a correct yes-or-no answer.

\begin{theorem}[Undecidability of Eventual Profit]
  \label{thm:profit-undecidable}
  Fix the positive computable price path \(P_t=t+1\), \(t=0,1,\ldots\).
  There is no algorithm which, under the promise that its input code computes
  a total trading program \(M_\sigma\), correctly decides for every such code
  whether
  \[
    \exists T\geq1:\qquad
    G_T(M_\sigma;P)=
    \sum_{t=0}^{T-1}w_t(P_{t+1}-P_t)>0.
  \]
\end{theorem}

\begin{proof}
Assume such a profit certifier exists. Given an arbitrary Turing machine
\(A\) and input \(x\), construct a trading program \(M_{A,x}\) as follows.
At date \(t\), simulate \(A(x)\) for \(t+1\) steps. If it has halted within
those steps, set \(w_t=1\); otherwise set \(w_t=0\). Every date requires only
a finite simulation, so \(M_{A,x}\) is total. Along \(P_t=t+1\), every long
position earns one unit in the next period. Hence \(G_T>0\) for some \(T\) if
and only if \(A(x)\) eventually halts. A certifier for eventual trading profit
would therefore decide the Halting Problem, which is impossible
\cite{turing-36}.
\end{proof}

This reduction, a conversion of one decision problem into another that
preserves its answer, can also be placed directly on the market side, which is
the form most closely connected to source-code event certification.

\begin{theorem}[Undecidability of a Market Event]
  \label{thm:market-undecidable}
  There is no algorithm which, under the promise that its input code computes
  a total positive-price generator \(G\), correctly decides for every such
  code whether its generated path \(P^G\) ever satisfies \(P_t^G>P_0^G\).
\end{theorem}

\begin{proof}
Given a Turing machine \(A\) and input \(x\), define the total generator by
\(P_0^{A,x}=1\) and, for \(t\geq1\),
\begin{equation}
  P_t^{A,x}
  =1+\mathbf{1}\{A(x)\text{ halts within }t\text{ steps}\},
  \qquad t\geq1 .
  \label{eq:market-halting}
\end{equation}
Here \(\mathbf{1}\{\cdot\}\) is the indicator: one when its condition holds
and zero otherwise. Each quoted price requires only a finite \(t\)-step
simulation, so the generator is total and computable. Its initial price is one, and it ever quotes
a higher price if and only if \(A(x)\) halts. A universal decider for the stated
market event would therefore decide the Halting Problem.
\end{proof}

Theorem~\ref{thm:market-undecidable} makes ``universal computation'' precise:
the admissible generators can emulate an arbitrary machine and expose its
outcome through a price event. No algorithm can then certify every future
event of every such market. Both reductions concern promise domains: the
putative decider must halt and be correct on every promised total input.
Positive instances are semidecidable: simulation confirms them when the event
occurs, but may run forever when it does not. There is no terminating general
procedure for certifying its perpetual absence.

This is not Gold's theorem. The reduction receives source code and must halt
with a yes-or-no answer; Gold's learner receives accumulating data and need
only stabilize eventually. Indeed, the paths in
Eq.~\eqref{eq:market-halting} are identifiable in the limit: guess the flat
path, then switch permanently if the price becomes two. No finite date,
however, certifies that the initial guess is final. Thus halting excludes
uniform certification, whereas Gold excludes identification in the limit for
the class of all binary total recursive time functions. Restricted model
classes may still be learnable
\cite{blum75blum,angluin:83,ad-91,2016-AlexanderKarlSvozil-ml,moore}.

Rice's theorem gives a related general limit: any property that depends only
on the function computed, holds for some partial computable functions, and
fails for others is undecidable for arbitrary programs~\cite{Rice-1953}.
Totality itself is undecidable in general~\cite{turing-36}. This does not
prevent evaluating a fixed total strategy on a completed finite computable
path by simulation. The explicit reductions above establish the restricted
claim needed here: future-event and eventual-profit questions remain
undecidable even when the input program is promised to be total.

\paragraph{Busy-Beaver horizons: finite glimpses and noncomputable waiting}

Fix a step-counted universal machine \(V\) with an effective compiler that adds
only fixed overhead for a given program schema, or template. The compiler
translates descriptions into code for \(V\); fixed overhead is a constant
number of extra code bits. Define the runtime Busy Beaver function by
\begin{equation}
  \operatorname{BB}(n)
  =
  \max\{\operatorname{time}_V(p):
  |p|\leq n\text{ and }V(p)\text{ halts}\},
  \label{eq:busybeaver}
\end{equation}
with value zero if the set is empty. Here \(p\) is a program, \(|p|\) its
length in bits, and \(\operatorname{time}_V(p)\) its number of execution steps
on \(V\). The input \(n\) is a program-description budget, not a date; \(\operatorname{BB}(n)\) is a runtime, which becomes a
possible trading date when the computation is embedded in a market generator.
A small increase in rule length may therefore require waiting beyond every
computably prescribed scale.

Each \(\operatorname{BB}(n)\) is finite, and selected low values can sometimes
be proved by exhaustive simulation plus proofs that the remaining programs do
not halt. This is not uniform ``computability for small arguments'': any finite
table can be hard-coded, and its apparent pattern licenses no extrapolation. If
a total computable \(h\) upper-bounded \(\operatorname{BB}\), simulating every
program of length at most \(n\) for \(h(n)\) steps would decide halting.
Eventual domination uses a further program-schema argument. For each \(n\), a
program of \(O(\log n)\) bits---at most a constant times \(\log n\) for
sufficiently large \(n\)---can encode \(n\), compute \(h(n)\), execute more
than \(h(n)\) computation steps solely to delay termination, and halt.
Here and below, \(\log\) denotes the natural logarithm. For all sufficiently large \(n\), its
compiled \(V\)-code has length at most \(n\), so
\(\operatorname{BB}(n)>h(n)\). Thus, up to the fixed machine convention,
\(\operatorname{BB}\) eventually exceeds every total computable
function~\cite{rado,chaitin-bb}.

Calude, Dinneen, and Shu illustrate the finite/infinite divide. For one explicit
universal machine \(U\) with prefix-free program codes, meaning no valid code
is a proper prefix of another, computation and mathematical proof settled
halting through 84-bit programs and certified 64 initial bits of its halting
probability \(\Omega_U\). This is the probability that \(U\) halts when its
program bits are supplied by independent fair coin tosses. They did not compute Busy Beaver values, and
the work needed to recover the relevant halting programs from longer supplied
prefixes of \(\Omega_U\) outruns every computable bound. Their result is a
formidable finite glimpse, not a scalable computation of an uncomputable object
\cite{2002-glimpseofran}.

Applied to Eq.~\eqref{eq:market-halting}, a short generator can quote
\(P_t=1\) for a Busy-Beaver-scale interval and then quote \(P_t=2\). Across
encoded programs of length at most \(n\), the last halt can occur at
\(\operatorname{BB}(n)\), up to fixed wrapper overhead. Knowing that horizon
would make continued silence a certificate of no later move. Equivalently,
with \(\mathrm{TotComp}\) denoting total computable size-dependent schedules,
\begin{equation}
  \begin{aligned}
  \forall h\in\mathrm{TotComp}\ \exists p:\quad
    &V(p)\text{ halts},\\
    &\operatorname{time}_V(p)>h(|p|).
  \end{aligned}
  \label{eq:bb-horizon}
\end{equation}
If \(c\) is the wrapper's fixed description-length overhead, it is absorbed by
the computable envelope \(H(n)=\max_{k\leq n+c}h(k)\). Every quote at a
specified finite date remains
computable by bounded simulation; the generator neither computes Busy Beaver
nor selects its maximizer uniformly. Noncomputability lies in certification
across the family. Each realized path is deterministic and computable, hence
not Martin-L\"of random under fair-coin coding.

This barrier is stronger than sensitive dependence on initial conditions or
visual complexity alone in a precise decision-theoretic sense. Repeatedly
applying finitely many distance-shrinking maps, or a chaotic map in which
nearby initial states can separate strongly, can generate intricate output
from a short computable rule. Turing universality additionally creates questions about
the unbounded future for which no terminating general algorithm exists. The
categories can overlap when a dynamical system performs universal computation.
This is a worst-case claim across generator descriptions, not evidence that
typical financial markets exhibit Busy-Beaver delays
\cite{barnsley:88,moore,calude:037103}.

These obstructions concern unbounded horizons. At any specified finite horizon,
total programs can still be run and their realized prices and gains computed.

\section{Worst-Case Value and Private Randomization}
\label{sec:randomization}

Returning to the direct countermarket, its finite-horizon value can be stated
explicitly, and the role of a trader's hidden randomization can be separated
from deterministic pathwise opposition. The maximin value is the best
worst-case lower bound on gain: the inner infimum ranges over markets, and the
outer supremum ranges over strategies.

\begin{corollary}[Maximin Value]
  \label{cor:maximin}
  At every finite horizon \(T\), if the neutral strategy (always holding zero
  shares) is allowed, then
  \[
    \sup_{\text{\rm total computable }M_\sigma}
    \inf_{P\in\mathcal P_{\mathrm{comp}}^{+}}
    G_T(M_\sigma,P)=0.
  \]
\end{corollary}

\begin{proof}
The neutral strategy guarantees zero, while
Theorem~\ref{thm:counterpath} gives every strategy a computable countermarket
with gain at most zero.
\end{proof}

\paragraph{Two sources of randomness}

Randomness can enter on either side: \emph{private trader randomization}
randomly selects positions, while \emph{market randomness} generates price
movements. They can occur separately or together and need not have the same
effect on profits. Private randomization need not destroy an existing edge,
and random market movements can have positive expected discounted returns.
The result below uses the stronger assumption that the next discounted price
change has conditional mean zero, even given the trader's private seed, the
random input used to choose positions. Thus the expected next change, given
all available information, is zero. Integrability means finite expected
absolute value; under this condition the gain has expectation zero at each fixed horizon,
not convergence of realized profits to zero in the long run.

\begin{proposition}[Private Randomization Does Not Create a Guarantee]
  \label{prop:randomized}
  Fix a finite horizon \(T\), a rational \(P_0>0\), and a rational
  \(\epsilon\in(0,1)\). There is a single passive computable market law,
  independent of the trader, under which every predictable randomized strategy,
  choosing positions using only information already available, with integrable
  gains has expected gain zero. Specifically, let
  \[
    P_{t+1}=P_t(1+\epsilon\xi_{t+1}),
    \qquad t=0,\ldots,T-1,
  \]
  where \((\xi_t)_{t\geq1}\) are independent fair signs, each equal to \(+1\)
  or \(-1\) with probability one half, also independent of the trader's private
  random seed. If \(w_t\) is measurable with respect to
  the information available when the position is chosen and
  \(w_t(P_{t+1}-P_t)\) is integrable, then
  \[
    G_n=\sum_{t=0}^{n-1}w_t(P_{t+1}-P_t),
    \qquad n=0,\ldots,T,
  \]
  is a martingale: its expected next value, conditional on current information,
  equals its current value. In particular, \(\mathbb E[G_T]=0\), where
  \(\mathbb E\) denotes expectation.
\end{proposition}

\begin{proof}
Let \(R\) denote the trader's private seed and
let \(\mathcal G_t\) be the sigma-algebra generated by
\(P_0,\ldots,P_t,R\): the collection of events determined by the
observed prices and the private seed. Allowing
\(\mathcal G_t\) to contain all of \(R\) only strengthens the claim.
Predictability makes \(w_t\) \(\mathcal G_t\)-measurable, meaning determined
by this information, while \(\xi_{t+1}\) is independent of
\(\mathcal G_t\) and has mean zero. Hence
\[
 \mathbb E\!\left[
   w_t(P_{t+1}-P_t)\mid\mathcal G_t
 \right]
 =\epsilon P_t w_t\,
   \mathbb E[\xi_{t+1}\mid\mathcal G_t]
 =0.
\]
Here \(\mathbb E[\,\cdot\mid\mathcal G_t]\) is expectation conditional on
that information. The integrability assumption therefore makes \((G_n)_{n=0}^T\) a martingale,
and \(\mathbb E[G_T]=\mathbb E[G_0]=0\).
\end{proof}

Private randomization prevents a deterministic countermarket based only on the
public history from mechanically opposing the unseen current draw, but it does
not yield strictly positive expected gain in every environment. More precisely,
under the proposition's integrability assumptions it excludes
\(\forall P:\mathbb P_R(G_T(P,R)>0)=1\) at the fixed horizon \(T\), where
\(\mathbb P_R\) takes probability over the trader's seed for a fixed path.
Each of the finitely many fair-sign prefixes has positive probability and a
computable infinite continuation; such a guarantee would imply
\(G_T>0\) almost surely jointly, hence \(\mathbb E[G_T]>0\), a contradiction.
It does not exclude a high win probability, profit under another law, or
benchmark-relative success.

\paragraph{Optional online-game reading}

The recursion in Theorem~\ref{thm:counterpath} also defines a finite
perfect-information game, in which both players observe every previous move:
the trader announces \(w_t\), the environment chooses
the next return, and Corollary~\ref{cor:maximin} gives the trader-first maximin
value zero. This reading is optional because the fixed deterministic trader's
countermarket program already generates the same path offline. Borel
determinacy concerns infinite games whose winning sets are generated from
finite-history events by complements and countable unions (Borel sets); it
does not by itself identify a winner here~\cite{Martin1975}. Hidden private randomization is
a different information structure: Proposition~\ref{prop:randomized} proves
zero expected gain under one passive computable fair-sign law, not a value
for a game with randomized strategies without specified strategy spaces.
Section~\ref{sec:countertrader} specifies a simultaneous game and
distinguishes its mixed-strategy equilibrium from this martingale statement.

\section{A Passive Market: Martin-L\"of Randomness}
\label{sec:randomness}

\paragraph{When does process randomness yield a random record?}

Process randomness describes how individual bits are generated and depend on
their past. Martin-L\"of (ML) randomness is a property of the resulting
\emph{infinite sequence} relative to a specified probability law.
\emph{Independent fair bits produce a fair-coin ML-random sequence with
probability one.} Equivalently, the sufficient process condition is
\begin{equation}
  \mathbb P(X_{n+1}=1\mid X_1,\ldots,X_n)=\tfrac12
  \quad\text{a.s. for every }n\geq0,
  \label{eq:conditional-fairness}
\end{equation}
Here \(X_n\) is the random bit at date \(n\), \(\mathbb P\) denotes
probability, and ``a.s.'' means almost surely, or with probability one. The
condition includes the first bit at \(n=0\).

Independence alone is insufficient for \emph{fair-coin} randomness: independent
bits with a fixed probability \(p\neq1/2\) of being one almost surely have
limiting frequency \(p\).
Fair one-bit probabilities alone are also insufficient: repeating one fair
bit forever gives fair marginals---each bit separately is equally likely to
be zero or one---but a constant, predictable record.
Independence is not necessary for ML randomness relative to the process's
\emph{own} computable law, which may include bias or dependence.

For FAPP use, the conclusion is therefore conditional on a justified source
model, for example independent fair trials. Tests for bias and serial dependence (statistical dependence between different
dates) can challenge that model over observed scales; passing them cannot
certify independence indefinitely or the ML randomness of an infinite output.
An almost-sure conclusion allows exceptional records of probability zero.
A fixed deterministic pseudorandom generator with a finite seed cannot produce
a fair-coin ML-random infinite sequence, however well its finite outputs test.

Here is the precise classical result. A computable law \(\mu\) has cylinder
probabilities \(\mu([s])\) uniformly approximable to any requested accuracy,
where \(s\) is a finite bit string and its cylinder \([s]\) comprises all
infinite sequences beginning with \(s\). An ML test is a uniformly computably
enumerable sequence of open sets \((U_k)\) with \(\mu(U_k)\leq2^{-k}\):
each set is a union of cylinders, and one algorithm lists their defining
prefixes given \(k\). A record is \(\mu\)-ML-random if it avoids
\(\bigcap_kU_k\), the set common to every level, for every such test. Write
\(\mathrm{MLR}_\mu\) for the set of these random records.

\begin{proposition}[Classical Measure-One Result, Martin-L\"of]
  \label{prop:process-ml}
  If a binary process has computable joint law \(\mu\), then
  \begin{equation}
    \mathbb P\bigl((X_1,X_2,\ldots)\in\mathrm{MLR}_\mu\bigr)=1.
    \label{eq:process-ml}
  \end{equation}
  This includes independent fair trials, independent bits with a fixed
  computable probability of one (a Bernoulli bias), and computable dependent laws~\cite[Sec.~IV]{MartinLof1966}.
\end{proposition}

\begin{proof}
Each test's failure set has measure zero. There are countably many effective
tests, so their union still has measure zero.
\end{proof}

The law must be specified: randomness relative to an atomic law, which gives
positive probability to individual records, need not mean fresh uncertainty
at each trial. For example, either constant record is random
for the law that repeats one fair bit forever.

\paragraph{Borel normality is weaker}
\label{par:normality}

A sequence is Borel normal in base \(b\) when every word, or contiguous block,
of \(k\) digits has limiting frequency \(b^{-k}\), for every \(k\geq1\);
the digit alphabet is \(0,\ldots,b-1\).
ML randomness for uniform base-\(b\) digits implies normality, but the converse
fails. Champernowne's decimal
\(0.123456789101112\ldots\), formed by concatenating the positive integers,
is normal in base ten yet computable, hence not ML-random for the uniform
decimal law~\cite{ch1}. Normality is difficult to establish for a given
infinite record: one must control all block lengths in the infinite limit.
No finite frequency check proves it, since any finite prefix admits both
normal and nonnormal continuations. A proof from a construction, as for
Champernowne, can establish normality; extensive sampling alone cannot.

\section{Trading, Martingales, and Algorithmic Randomness}
\label{sec:martingales}

Effective randomness can also be expressed through the inability of specified
betting strategies to accumulate unbounded capital. The computational class
of the strategy matters: computable randomness and Martin-L\"of randomness
use different classes~\cite{Schnorr1971,DH,Nies2009}.
This betting viewpoint complements Levin's characterization using the
monotone description complexity of initial segments~\cite{levin-73}.

Fix rational \(P_0>0\), rational \(0<\epsilon<1\), and initial credit \(C>0\).
For \(x=(x_1,x_2,\ldots)\in\{0,1\}^{\mathbb N}\), set
\[
  P_t(x)=P_0\prod_{i=1}^{t}(1+\epsilon(2x_i-1)).
\]
Each price history in this class determines its bit string.
A trader is \emph{admissible with credit \(C\)} if its gain satisfies
\(G_n\ge-C\) at every date on every such history.
Thus \(C+G_n\) is nonnegative wealth. A nonnegative function \(d\) on finite
binary strings is a \emph{martingale} if
\[
  d(s)=\tfrac12\bigl(d(s0)+d(s1)\bigr).
\]
It \emph{succeeds} on \(x\) when
\(\limsup_n d(x\upharpoonright n)=\infty\), where
\(x\upharpoonright n\) is the first \(n\) bits.

\begin{proposition}[Traders and Martingales]
\label{prop:trader-martingale}
Total computable traders returning exact rational positions and admissible with rational credit
\(C\) correspond exactly, on this price class, to nonnegative martingales
with exactly computed rational values and \(d(\emptyset)=C\).
The correspondence and its inverse are
\[
  d(s)=C+G_{|s|}(s),\qquad
  w_{|s|}(s)=\frac{d(s1)-d(s0)}{2\epsilon P_{|s|}(s)}.
\]
\end{proposition}

\begin{proof}
The two next gains are \(\pm w_{|s|}(s)\epsilon P_{|s|}(s)\), so their
average is zero. Conversely, the displayed position gives gain
\(d(sb)-d(s)\) on extension by \(b\); these increments telescope.
Nonnegativity is precisely admissibility. All quantities are rational and
computable. Set the trader's position to zero on price histories outside
the specified class, including the empty history, to obtain a total rule.
\end{proof}

\begin{theorem}[Effective Betting Characterizations]
\label{thm:ville}
A binary record \(x\) is computably random if and only if every total
computable trader returning exact rational positions and admissible with some rational credit
has bounded gain along \(P(x)\).
It is Martin-L\"of random if and only if no nonnegative
left-computably-enumerable martingale succeeds on it.
\end{theorem}

Here \emph{left-computably-enumerable}, or \emph{lower semicomputable}, means
that the values have uniformly computable increasing rational approximations,
without a required computable error bound.
The standard betting characterizations give the theorem
\cite{Schnorr1971,DH,Nies2009}.
Restricting computable martingales to exact rational outputs does not change
computable randomness. Given a computable-real martingale \(d\), choose
rational initial capital greater than \(d(\emptyset)+1\). Approximate each
capital wager \((d(s1)-d(s0))/2\) by a rational within
\(2^{-|s|-2}\), and update the new capital by adding or subtracting that
rational. Along every path the accumulated error is less than \(1/2\).
The resulting rational martingale stays nonnegative and differs from \(d\)
by a bounded amount, so preserves success.

The second characterization must not be read as an equivalence with our
executable rational traders. A left-computably-enumerable martingale may
have noncomputable initial capital. Indeed, if its initial capital is
computable, its averaging identities and lower approximations give upper
approximations too, making all its values computable reals. Furthermore,
the difference \(d(s1)-d(s0)\) in the position formula need not be lower
semicomputable or rational. Lower approximations to wealth therefore do not
by themselves produce computable orders.

There is nevertheless one left-computably-enumerable martingale succeeding
on every non-Martin-L\"of-random record. To see this, take a universal
Martin-L\"of test \((U_n)_{n\ge1}\), with fair-coin measure
\(\mu(U_n)\le2^{-n}\), and put
\[
  D(s)=\sum_{n\ge1}\frac{\mu(U_n\cap[s])}{\mu([s])}.
\]
Here \([s]\) is the set of infinite extensions of \(s\).
The sum is finite at every string, is lower semicomputable, and satisfies
the martingale identity, with \(D(\emptyset)\le1\).
If \(x\in\bigcap_nU_n\), each of the first arbitrarily many summands
eventually equals one along \(x\), proving success.
Universality here concerns successful records; it does not assert
multiplicative domination of every left-computably-enumerable martingale.
In particular, it is not obtained by effectively listing all such
martingales and taking their weighted sum~\cite{DH,Nies2009}.

For any specified computable record and rational \(C>0\), an ordinary rational trader can instead
stake its available wealth on each computed next sign. Its wealth is
\(C2^n\) along that record and becomes zero after an incorrect sign.
Thus a computable countermarket is exploitable without bound by another
admissible trader. For a known rule \(\Phi\) as in
Proposition~\ref{prop:countertrader}, compute its return from the actual
prior positions before each stake; the same doubling follows.
The trader still depends on the specified record or rule. The universal process above uses
a different computational class, and success concerns infinite-time
unboundedness, not guaranteed profit at a fixed horizon.

\paragraph{A probability bound}
Under independent fair signs, every predictable admissible trader, including
one using an independent private random seed, satisfies Ville's bound
\cite{Ville1939}:
\[
  \mathbb P\!\left(\max_{0\le t\le T}G_t\ge K\right)
  \le\frac{C}{C+K},\qquad K>0.
\]
For randomized traders, require measurable positions and admissibility on
every history for each seed. Nonnegativity of both next wealth values gives
\(\epsilon P_t|w_t|\le C+G_t\), hence \(0\le C+G_t\le2^tC\).
Thus finite-horizon integrability follows automatically.
Stopping this martingale when it first reaches \(C+K\), or at \(T\),
gives \(C\ge(C+K)\mathbb P(\max_{t\le T}G_t\ge K)\).
This quantifies Proposition~\ref{prop:randomized} without requiring
computability of the trader.

\section{Ramsey Regularity and the Finite-Pattern Trap}
\label{sec:ramseytrap}

A large market record can be pattern-rich even when it contains no exploitable
law. Ramsey theory makes this finite combinatorial point without assuming a
probability model. For fixed positive integers \(c,m,q\), there is a finite
number \(R_q(m;c)\) such that every assignment of one of \(c\) colors to
each \(q\)-element subset of any set of size at least \(R_q(m;c)\) has an
\(m\)-element subset whose \(q\)-subsets all receive the same color. Thus every red--blue coloring of the
pairs among six objects contains a monochromatic triangle, with all three
pairs the same color. Van der Waerden's theorem similarly forces an arithmetic
progression---equally spaced indices---of any prescribed finite length, all
of one color, in every sufficiently long finite coloring
\cite{Ramsey-GRS-90,GS-90}. These thresholds may be enormous, but they are
finite and computable, unlike a Busy-Beaver horizon.

The guarantee depends on the coding and is weaker than a trading signal. In a
discretized return record, equal symbols at equally spaced dates need not be
consecutive, form a price trend, have a favorable sign, or persist beyond the
sample. Ramsey existence therefore asserts neither statistical dependence nor
causation, predictability, or profit. Nor does it conflict with algorithmic
randomness. A fair-coin Martin-L\"of-random sequence is Borel normal, so every
finite binary word, including arbitrarily long runs and finite chart motifs,
occurs with its fair-coin limiting frequency. Randomness here means escaping
every effective null test, not local patternlessness
\cite{calude:02,Nies2009}.

Calude and Longo call regularities forced by data volume \emph{Ramsey-type
correlations}. Their incompressibility argument emphasizes that almost all
sufficiently long finite strings resist compression by any fixed positive
fraction while still containing the finite regularities forced by Ramsey
theory~\cite{Calude2016}. This is an epistemic warning, not a theorem that every
sample correlation is false: a Ramsey configuration is not itself a
correlation coefficient, and a pattern observed in one finite history is
compatible with both random and computable continuations.

Trading research compounds this pattern abundance with selection. A null
hypothesis is the baseline claim assessed by a statistical test, such as the
absence of a forecasting relation. If all \(M\) null hypotheses are true,
their rejection events are independent, and
each has rejection probability exactly \(\alpha\), the probability of at least
one false rejection is
\(1-(1-\alpha)^M\); if their probabilities are at most \(\alpha\), this is an
upper bound. Without independence, the union bound, which bounds the
probability of any rejection by the sum of individual rejection probabilities,
gives
\(\min\{1,M\alpha\}\) as an upper bound. This
probabilistic multiple-testing problem is distinct from Ramsey's deterministic
existence theorem, but its practical lesson is aligned: a trend or cross-market
relation retained after searching many assets, windows, transformations, lags,
and parameters is a candidate hypothesis, not yet an edge. A defensible FAPP
claim must report the search universe and survive selection correction
(accounting for the search over alternatives), costs, an untouched holdout
(data never used to choose or tune the strategy) or independent replication,
and testing outside
the selected regime~\cite{BaileyLopez2014}.

\section{Countertraders and Adaptive Mutual Countering}
\label{sec:countertrader}

\begin{definition}[Countertrader]
For a specified market and initial price, a \emph{countertrader} is a total
deterministic computable strategy satisfying
\[
  w_t(P_{t+1}-P_t)>0
\]
at every date along its interaction with that market. The strategy may
depend on the market's program; the definition does not require discovering
an unknown program from observations.
\end{definition}

Every fixed positive rational-valued computable price path with a nonzero
change at every date admits such a trader. Given its generating program,
compute the next change and set
\(w_t=\operatorname{sgn}(P_{t+1}-P_t)\), where
\(\operatorname{sgn}\) returns \(+1\) for a positive argument and \(-1\) for a
negative one. The gain is \(|P_{t+1}-P_t|>0\).
This is still a strategy of the form used in
Section~\ref{sec:countermarket}: it can use the supplied history's length to
find the date, compute the specified path, and ignore the observed price
values; on the unused empty history, return zero. Its guarantee concerns
that path only, and relies on the costless
trading and unrestricted computation time assumed above.

A similar construction applies to a market that reacts to earlier positions.
The restriction to \emph{earlier} positions specifies the information available
when the current return is computed.

\begin{proposition}[Countertrader for a Computable Market Rule]
\label{prop:countertrader}
Fix rational \(P_0>0\) and rational \(0<\epsilon<1\).
Let a total computable function
\(\Phi\colon\mathbb Q^{<\omega}\to\{-\epsilon,+\epsilon\}\)
determine the next proportional return by
\[
  r_t=\Phi(w_0,\ldots,w_{t-1}),\qquad
  P_{t+1}=P_t(1+r_t).
\]
Then this market admits a countertrader, with strictly positive gain at
every finite horizon \(T\geq1\).
\end{proposition}

\begin{proof}
Recursively choose
\(w_t=\operatorname{sgn}\Phi(w_0,\ldots,w_{t-1})\).
Only previously computed positions are needed, so totality of \(\Phi\)
makes the position sequence computable. Its gain in period \(t\) is
\(w_tP_tr_t=\epsilon P_t>0\); summing gives the horizon claim.
To obtain a total function on arbitrary nonempty price histories, use their
length to regenerate this position sequence; return zero on the empty history.
\end{proof}

This does not contradict Theorem~\ref{thm:counterpath}.
That theorem fixes a trader \(M_\sigma\) and constructs its countermarket
\(P^{M_\sigma}\). A countertrader exploiting this path is a different trader.
Constructing a countermarket against the replacement can produce a different
path again. Thus every fixed trader has a countermarket, while every fixed
computable path with nonzero changes has a countertrader.
Each responding program depends on the opponent already specified.
The quantifiers express dependence, not the temporal order of deployed trades.
The success criteria also differ: the countertrader earns strictly positive
gain, whereas the countermarket prevents strictly positive gain and cannot
make an abstaining trader lose. Allowing flat price paths would invalidate
the stated countertrader guarantee.

\paragraph{Perfect current responses}

Write \(a_t\in\{-1,+1\}\) for a unit position and
\(d_t\in\{-1,+1\}\) for the return direction.
A perfectly profitable trader and a perfectly opposing market would require,
respectively,
\[
  a_t=d_t,\qquad d_t=-a_t.
\]
These equations have no solution. This rules out their simultaneous
satisfaction, not interaction between arbitrary deterministic programs:
any two total rules on past histories produce well-defined play.
If each program first simulates the other on the same current history and
returns only after that simulation returns, neither computation terminates.
This violates the totality assumption
\cite{FallensteinTaylorChristiano2015}.

If the trader announces its current position before the market chooses the
return, the market can oppose every nonzero position; abstention secures zero.
If the market commits to a nonzero current return and reveals it before the
position is chosen, the trader can profit. Knowledge of a deterministic
opponent's code can sometimes replace observation of its action, but does not
make incompatible perfect-response requirements consistent.
Calling a market rational also requires assigning it an objective, such as
minimizing this trader's gain; computability alone supplies no such objective.

Total computability supplies no trading deadline. A simulation may finish
too late. In the mutually dependent current-response construction, timeouts
or fallback actions can restore termination only by abandoning at least one
exact countering requirement. There is no general rule that a counterstrategy needs
a larger computational complexity class: an always-up market is countered
by an always-long trader. Equal computation budgets alone determine neither
the winner nor a zero payoff.

\paragraph{Hidden random choices}

For a simple one-period simultaneous game, fix the current price \(P_t>0\),
restrict positions to \(w_t\in\{-1,0,+1\}\), retain directions
\(d_t\in\{-1,+1\}\), and assign the trader payoff
\(\epsilon P_tw_td_t\) and the market its negative.
This is matching pennies with an abstention option. Its value is zero.
A mixed equilibrium, meaning neither side can improve its expected payoff
by changing its own choice distribution alone, has a fair market direction
and equal trader probabilities for long and short. Any remaining probability
may be assigned to abstention. Choices are hidden and independently
randomized. The stated distributions give expected payoff zero against every
opposing pure action, which proves the equilibrium claim.
There is no pure equilibrium: any fixed nonzero direction can be exploited,
and any active fixed position can be opposed.

This equilibrium statement is distinct from
Proposition~\ref{prop:randomized}. That proposition makes every predictable
trader's gain, subject to its integrability assumptions, a martingale under
the specified fair-sign market law. It does not require mean-zero positions
or establish mutual best responses. For example, always holding one share
has expected gain zero under that law, although an opposing market would
prefer downward moves.

Reflective-oracle models provide a more general framework for consistent
mutual reasoning through oracle access and probabilistic answers
\cite{FallensteinTaylorChristiano2015}. This changes the computational model;
ordinary randomization does not make exact universal next-move prediction possible.
Related tensions between prediction and optimization arise in repeated-game
learning under assumptions on beliefs or uncertainty about opponents' payoffs
\cite{Nachbar1997,FosterYoung2001}. These results provide context; the
countertrader construction and the incompatible current-response equations
already establish the distinction needed here.

\section{Practical Robustness Diagnostics: Time Reversal and Permutations}
\label{sec:criterion}

This section turns the trader--market leitmotif into a falsification tool. Hold
the trader fixed and transform only the market record. Among the countable
family of computably specified stress transformations, time reversal is one of
the simplest, particularly for detecting dependence on directional
persistence. It is a diagnostic, rather than a further impossibility theorem
like the countermarket construction or Halting reductions.

Let \(\mathcal{M}\) be a space of strictly positive discounted price paths
\(P\colon[0,T]\to\mathbb{R}_{>0}^n\), where \(\mathbb{R}_{>0}^n\) denotes
vectors of \(n\) positive real prices, one per asset, and \(T\) is the horizon. The time reversal of \(P\) is
\begin{equation}
  \widetilde{P}(t)=P(T-t), \qquad t\in[0,T].
  \label{eq:timereversal}
\end{equation}
Literal reversal reuses the observed price levels and preserves the magnitudes
of log-price moves, the logarithms of successive price ratios, while reversing
their order and signs. It is not the same as merely negating arithmetic
returns: a gross return \(R\) is the later price divided by the earlier price,
and the corresponding arithmetic return is \(R-1\). The reversed gross return
is \(1/R\). In an empirical implementation
the strategy must be rerun causally from scratch on the reversed record,
using only observations already available at each decision. Reversing the
original orders or realized profit would leak future information.
If the initial level is prescribed, literal reversal can leave the market
class. One may instead test componentwise normalized reversal,
\[
  P_i^{\mathrm{rev}}(t)=P_i(0)\frac{P_i(T-t)}{P_i(T)},
\]
where \(i\) labels the asset and \(\mathrm{rev}\) its reversed path,
provided it remains in the claimed class. Scaling can change the actions of a
strategy that is not scale invariant, meaning that its decisions change when
all prices of an asset are multiplied by the same positive constant; it is
therefore a distinct test.
For a causal strategy \(\sigma\), let \(\Pi_t[\sigma;P]\) denote its discounted
gain from zero initial capital through time \(t\) on \(P\). Call \(\sigma\)
\emph{pathwise admissible} on \(\mathcal M\) if some finite \(C\) satisfies
\(\Pi_t[\sigma;P]\geq-C\) for every \(P\in\mathcal M\) and
\(t\in[0,T]\).

\begin{definition}[Pathwise Universal Strategy]
\label{def:universal}
  A causal strategy \(\sigma\) is \emph{pathwise universal} on
  \(\mathcal{M}\) if it is self-financing and pathwise admissible there and
  satisfies \(\Pi_T[\sigma;P]>0\) for every \(P\in\mathcal{M}\).
\end{definition}

\begin{criterion}[Reversal Falsification Test]
  \label{crit:tr}
  Suppose \(\mathcal{M}\) is closed under time reversal, meaning that reversing
  any member produces another member. If there exists \(P\in\mathcal{M}\) such that
  \begin{equation}
    \min\{\Pi_T[\sigma;P],\Pi_T[\sigma;\widetilde P]\}\leq 0,
    \label{eq:reversaltest}
  \end{equation}
  then \(\sigma\) is not pathwise universal on \(\mathcal{M}\).
\end{criterion}

\begin{proof}
Both \(P\) and \(\widetilde P\) belong to \(\mathcal{M}\). A non-positive gain
on either member of the pair contradicts Definition~\ref{def:universal}.
\end{proof}

One non-winning path refutes a claim quantified over every path. Reversal
supplies a paired candidate, not a proof that one member must fail: a causal
strategy may take different positions on the reversed history.
The countermarket construction provides the impossibility result for total
deterministic computable traders; reversal is a practical probe of a proposed
strategy.

\paragraph{Permuting a finite record}

Time reversal is a special permutation of a sampled record. For a finite
record of discounted price levels \(P=(P_0,\ldots,P_N)\), let \(\pi\) be
a permutation of \(\{0,\ldots,N\}\), a rearrangement using every index exactly
once, and define
\begin{equation}
  (P^\pi)_t=P_{\pi(t)},\qquad t=0,\ldots,N.
  \label{eq:price-permutation}
\end{equation}
Thus \(\pi(t)=N-t\) gives reversal. A fixed algorithm or a random shuffle
can choose \(\pi\). For vector-valued prices, each complete price vector
moves together. Let \(\mathcal M_N\) be the claimed class of such discrete
records and \(\Pi_N[\sigma;P]\) the discounted gain from zero initial capital
when the same causal rule is rerun from its prescribed initial state on \(P\).
Use the discrete version of Definition~\ref{def:universal} on this class.

\begin{criterion}[Permutation Falsification Test]
  \label{crit:permutation}
  A pathwise universal strategy on \(\mathcal M_N\) must satisfy
  \begin{equation}
    \begin{aligned}
      &\Pi_N[\sigma;P^\pi]>0\\
      &\quad\text{whenever }P\in\mathcal M_N\text{ and }P^\pi\in\mathcal M_N.
    \end{aligned}
    \label{eq:permutationtest}
  \end{equation}
  One permitted permutation with zero or negative gain refutes that universal
  claim. The criterion requires positivity, not equality of profits across
  permutations.
\end{criterion}

\begin{proof}
Every retained \(P^\pi\) is a member of the claimed market class, so strict
profit on every member requires strict profit on this one.
\end{proof}

The shuffle rearranges price levels, not returns or previously chosen orders.
At each date the rerun uses only the history available in the new order.
Permuting levels preserves their multiset, the collection of levels with
repetitions counted, but generally changes the returns and their magnitudes. It can also violate a fixed initial-price condition,
a bound on price moves, or another restriction defining \(\mathcal M_N\).
Only transformations that remain within that class test its universal claim;
closure under all permutations is unnecessary. Failure after destroying a
trend outside a strategy's stated domain does not refute a claim confined to
that domain. Passing finitely many shuffles of one record does not establish
universal profitability on a larger path class.

This diagnostic assumes no uniform distribution over records or equality of
average performance. Randomness in the choice of stress records
alone supplies neither a market law nor a statistical significance level.

Greg Egan's \emph{Permutation City} (1994)~\cite{permutationcity} provides a
literary analogy by exploring the temporal ordering of simulated computational
states. Here permutations change the trader's observation sequence; the causal
rule is rerun on each record, and Criterion~\ref{crit:permutation} follows from
the universal-profit requirement.

\begin{remark}[Computable Stress Tests Beyond Reversal]
\label{rem:stressfamily}
Time reversal is one of countably many computable probes: reversing return
signs, joining segments from different market conditions (regime splicing),
adding price jumps or execution delays, changing the size of price
fluctuations (volatility rescaling), and varying financing or transaction costs. Since program totality is undecidable, an operational battery must select
known-terminating transformations whose outputs remain in the claimed market
class. Here ``test'' means empirical stress, not a Martin-L\"of test: a uniformly
computably enumerable sequence of open sets with effective measure bounds.
Reversal defines no null set and certifies no randomness; passing it shows only
tested robustness, not universality.
\end{remark}

\begin{remark}[Filtrations and Time Reversal]
Stochastic model classes are not automatically reversal-closed:
The forward filtration is the information accumulated in the original time
order. The reversed value \(P(T-t)\) is generally not adapted to it, meaning
that this value need not be determined by information available at time \(t\). Reversing a diffusion, a continuous process driven by Brownian noise,
requires additional hypotheses~\cite{HaussmannPardoux1986}. For
example, if
\[
  X_t=\log(S_t/S_0)=mt+\nu W_t,
\]
\(X_t\) is the stock's cumulative log return, \(m\) its drift (average rate of
change), and \(W_t\) standard Brownian motion, whose independent normal
increments have mean zero and variance equal to elapsed time. Here
\(\nu\geq0\) denotes volatility, the scale of random fluctuations. Then \(X_{T-t}-X_T\), under the information revealed
by the reversed increments, has drift \(-m\) and volatility \(\nu\).
Thus a class containing both drift signs is closed in law under this
transformation: the transformed process has a distribution within the class.
Under the physical, or real-world, probability law, \(m=\mu-\nu^2/2\), where \(\mu\) is the stock's
instantaneous expected proportional return. Under a risk-neutral measure,
a probability law used for pricing through discounted expectations, the
discounted stock may be a martingale while its logarithm retains the
It\^o drift \(-\nu^2/2\), the correction caused by taking logarithms of
Brownian-driven prices. Zero drift in an additive Brownian price model does
not establish detailed balance in general, meaning invariance of equilibrium
history probabilities under time reversal. The diagnostic remains pathwise.
\end{remark}

\section{Conclusion}
\label{sec:conclusion}

Everything returns to the same pair: a trader and a market. For every fixed
total deterministic computable trader, Theorem~\ref{thm:counterpath} constructs
a fixed positive computable price path whose next bounded proportional move
opposes every nonzero position and remains nonzero during inactivity.
The quantifiers are
\(\forall\sigma\,\exists\mathcal E_\sigma\): no physical market must observe
the deployed orders, the always-flat trader earns zero, and every nonzero
position loses on its trader-relative countermarket. This direct diagonal
result uses no Halting-Problem reduction and identifies no single path that
defeats every trader.

The converse construction is also opponent-specific. A trader supplied
with a total computable program for nonzero price changes can choose
profitable positions, but remains vulnerable to a new countermarket.
These separate constructions do not yield mutually perfect current-move
responses: their interaction requires an explicit account of timing,
information, and termination (Section~\ref{sec:countertrader}).

The remaining limits concern different market sides. Gold precludes identifying
every computable binary generator in the limit, while next-bit diagonalization
defeats each fixed computable predictor on its predictor-relative record.
Halting reductions preclude uniform certification of eventual profit or encoded
future events, and Busy Beaver excludes a computable description-size waiting
horizon despite finitely solvable low-complexity cases. Under a stated
computable reference measure, the generated record is Martin-L\"of random
almost surely; neither this property nor Borel normality can be certified
from a finite sample. Ramsey regularity explains why finite
structured subsequences nevertheless remain abundant, including in random
records. These conclusions are complementary because their inputs,
quantifiers, and success criteria differ.

Proposition~\ref{prop:real-cumulative} preserves cumulative nonpositivity
even for approximation-based computable-real positions, although a computable
path need not make every active position lose. The martingale correspondence
in Section~\ref{sec:martingales} concerns bounded-credit trading and
unbounded wealth over an infinite record. A universal lower-semicomputable
capital process therefore does not provide an executable universal trader
with strict profit at a prescribed finite horizon.

These limits leave room for FAPP profitability under specified market
conditions, lawful information, or compensation for risk. Such claims remain
conditional and must be evaluated after costs and financing, against an
explicit benchmark. The selection controls in Section~\ref{sec:ramseytrap}
and the possibility that deployment changes the market remain essential;
neither establishes a universal profit guarantee.

Any repeatable-success claim must therefore state its market restriction,
probability law, benchmark, admissibility and financing conditions, horizon,
data protocol, and deployment effects. Time reversal and other within-class
permutations are robustness probes: one nonpositive gain refutes an all-path
claim, while passing a finite set of probes does not certify universality.
Useful strategies may exist, but no total deterministic computable
self-financing rule with exact rational-valued positions guarantees strict
gain on every positive rational-valued computable price path.

\begin{acknowledgments}
This research was funded in whole or in part by the Austrian Science Fund
(FWF), Grant DOI: \href{https://doi.org/10.55776/PIN5424624}{10.55776/PIN5424624}. For open access purposes, the author
has applied a CC BY public copyright license to any author accepted manuscript
version arising from this submission. The author acknowledges TU Wien
Bibliothek for financial support through its Open Access Funding Programme.

OpenAI Codex was used to assist with literature organization, LaTeX
editing, and checks of derivations. The author specified the scientific
direction and prompts; model output retained in the manuscript was checked
against the cited sources and explicit calculations. The author accepts full
responsibility for the final text.
\end{acknowledgments}

\section*{Data Availability}
This is a purely mathematical work, and no data or software were created or
analyzed in this study.


\bibliography{svozil}
\end{document}